\pdfoutput=2
\documentclass[conference, a4paper]{IEEEtran}
\IEEEoverridecommandlockouts
\usepackage[colorlinks=true,bookmarks=false,citecolor=black,urlcolor=black,linkcolor=black]{hyperref} 
\usepackage[hyphenbreaks]{breakurl}

\usepackage{cite}
\usepackage{amsmath,amssymb,amsfonts}
\usepackage{algorithmic}
\usepackage{graphicx}
\usepackage{textcomp}
\usepackage{booktabs}
\usepackage{xcolor}
\newcommand{\notimplies}{\;\not\!\!\!\implies}
\usepackage{tikz}
\usepackage{pgfplots}
\usepgfplotslibrary{groupplots}
\usepackage[nolist]{acronym} 
\usepackage{bm}

\definecolor{mittelblau}{RGB}{0, 126, 198}
\definecolor{violettblau}{cmyk}{0.9, 0.6, 0, 0}
\definecolor{rot}{RGB}{238, 28 35}
\definecolor{apfelgruen}{RGB}{140, 198, 62}
\definecolor{gelb}{RGB}{255, 229, 0}
\definecolor{orange}{RGB}{244, 111, 33}
\definecolor{pink}{RGB}{237, 0, 140}
\definecolor{lila}{RGB}{128, 10, 145}
\definecolor{hellgrau}{RGB}{224, 224, 224}
\definecolor{mittelgrau}{RGB}{128, 128, 128}
\definecolor{dunkelgrau}{RGB}{80,80,80}
\definecolor{anthrazit}{RGB}{19, 31, 31}
\definecolor{darkgreen}{RGB}{34,139,34}
\definecolor{aqua}{RGB}{0, 255, 255}
\definecolor{kit-green100}{rgb}{0,.59,.51}
\definecolor{kit-green70}{rgb}{.3,.71,.65}
\definecolor{kit-green50}{rgb}{.50,.79,.75}
\definecolor{kit-green30}{rgb}{.69,.87,.85}
\definecolor{kit-green15}{rgb}{.85,.93,.93}
\definecolor{KITgreen}{rgb}{0,.59,.51}

\definecolor{KITpalegreen}{RGB}{130,190,60}
\colorlet{kit-maigreen100}{KITpalegreen}
\colorlet{kit-maigreen70}{KITpalegreen!70}
\colorlet{kit-maigreen50}{KITpalegreen!50}
\colorlet{kit-maigreen30}{KITpalegreen!30}
\colorlet{kit-maigreen15}{KITpalegreen!15}

\definecolor{KITblue}{rgb}{.27,.39,.66}
\definecolor{kit-blue100}{rgb}{.27,.39,.67}
\definecolor{kit-blue70}{rgb}{.49,.57,.76}
\definecolor{kit-blue50}{rgb}{.64,.69,.83}
\definecolor{kit-blue30}{rgb}{.78,.82,.9}
\definecolor{kit-blue15}{rgb}{.89,.91,.95}

\definecolor{KITyellow}{rgb}{.98,.89,0}
\definecolor{kit-yellow100}{cmyk}{0,.05,1,0}
\definecolor{kit-yellow70}{cmyk}{0,.035,.7,0}
\definecolor{kit-yellow50}{cmyk}{0,.025,.5,0}
\definecolor{kit-yellow30}{cmyk}{0,.015,.3,0}
\definecolor{kit-yellow15}{cmyk}{0,.0075,.15,0}

\definecolor{KITorange}{rgb}{.87,.60,.10}
\definecolor{kit-orange100}{cmyk}{0,.45,1,0}
\definecolor{kit-orange70}{cmyk}{0,.315,.7,0}
\definecolor{kit-orange50}{cmyk}{0,.225,.5,0}
\definecolor{kit-orange30}{cmyk}{0,.135,.3,0}
\definecolor{kit-orange15}{cmyk}{0,.0675,.15,0}

\definecolor{KITred}{rgb}{.63,.13,.13}
\definecolor{kit-red100}{cmyk}{.25,1,1,0}
\definecolor{kit-red70}{cmyk}{.175,.7,.7,0}
\definecolor{kit-red50}{cmyk}{.125,.5,.5,0}
\definecolor{kit-red30}{cmyk}{.075,.3,.3,0}
\definecolor{kit-red15}{cmyk}{.0375,.15,.15,0}

\definecolor{KITpurple}{RGB}{160,0,120}
\colorlet{kit-purple100}{KITpurple}
\colorlet{kit-purple70}{KITpurple!70}
\colorlet{kit-purple50}{KITpurple!50}
\colorlet{kit-purple30}{KITpurple!30}
\colorlet{kit-purple15}{KITpurple!15}

\definecolor{KITcyanblue}{RGB}{80,170,230}
\colorlet{kit-cyanblue100}{KITcyanblue}
\colorlet{kit-cyanblue70}{KITcyanblue!70}
\colorlet{kit-cyanblue50}{KITcyanblue!50}
\colorlet{kit-cyanblue30}{KITcyanblue!30}
\colorlet{kit-cyanblue15}{KITcyanblue!15}
\usepackage{amsthm} %
\usetikzlibrary{pgfplots.groupplots}
\usetikzlibrary{shapes.geometric}
\usetikzlibrary{spy}
\usepackage{mleftright}
\usepackage{mathtools} %
\newcommand\blfootnote[1]{%
    \begingroup
    \renewcommand\thefootnote{}\footnote{#1}%
    \addtocounter{footnote}{-1}%
    \endgroup
}
\usepackage{etoolbox} %

\newtheorem{theorem}{Theorem}

\newcommand\orangetri{\resizebox{0.2cm}{!}{\tikz{\node[isosceles triangle, fill=KITorange, draw=KITorange, ultra thick] (A) {};}}}
\newcommand\blacktri{\resizebox{0.2cm}{!}{\tikz{\node[isosceles triangle, fill=black, draw=black] (A) {};}}}
\newcommand\orangecirc{\resizebox{0.2cm}{!}{\tikz{\node[circle, fill=KITorange, draw=KITorange, ultra thick] (A) {};}}}
\newcommand\blackcirc{\resizebox{0.2cm}{!}{\tikz{\node[circle, fill=black, draw=black] (A) {};}}}

\begin{document}

\title{Subcode Ensemble Decoding of Polar Codes\\}

\author{%
  \IEEEauthorblockN{%
    Henning Lulei\textsuperscript{\dag},
    Jonathan Mandelbaum\textsuperscript{\dag},
    Marvin Rübenacke\textsuperscript{\ddag},
    Holger Jäkel\textsuperscript{\dag},
  }
   \IEEEauthorblockN{%
    Stephan ten Brink\textsuperscript{\ddag},
    and Laurent Schmalen\textsuperscript{\dag}
  }
  \IEEEauthorblockA{%
    \textsuperscript{\dag}Karlsruhe Institute of Technology (KIT), Communications Engineering Lab (CEL), 76187 Karlsruhe, Germany\\
    \textsuperscript{\ddag}University of Stuttgart, Institute of Telecommunications, 70569 Stuttgart, Germany\\
    Email: {\texttt{jonathan.mandelbaum@kit.edu}
  }
}
}
\makeatletter
\patchcmd{\@maketitle}  %
{\addvspace{0.5\baselineskip}\egroup}
{\addvspace{-0.8\baselineskip}\egroup}
{}
{}
\makeatother

\maketitle

\begin{acronym}
    \acro{AED}[AED]{automorphism ensemble decoding}
    \acro{BI-AWGN}[BI-AWGN]{binary input additive white Gaussian noise}
    \acro{BDMC}[BDMC]{binary discrete memoryless channel}
    \acro{BER}[BER]{bit error rate}
    \acro{BP}[BP]{belief propagation}
    \acro{CRC}[CRC]{cyclic redundancy check}
    \acro{DE}[DE]{density evolution}
    \acro{FER}[FER]{frame error rate}
    \acro{LLR}[LLR]{log-likelihood ratio}
    \acro{LDPC}[LDPC]{low-density parity-check}
    \acro{ML}[ML]{maximum likelihood}
    \acro{PT}[PT]{pre-transformation}
    \acro{SC}[SC]{successive cancellation}
    \acro{SCED}[ScED]{subcode ensemble decoding}
    \acro{SCL}[SCL]{successive cancellation list}
    \acro{URP}[URP]{undecodable received pattern}
\end{acronym}
\newcommand{\PTA}{\mbox{\ac{PT}-A}}
\newcommand{\PTB}{\mbox{\ac{PT}-B}}
\newcommand{\PTC}{\mbox{\ac{PT}-C}}

\begin{abstract}
In the short block length regime, pre-transformed polar codes together with \ac{SCL} decoding possess excellent error correction capabilities.
However, in practice, the list size is limited due to the suboptimal scaling of the required area in hardware implementations. 
\Ac{AED} can improve performance for a fixed list size by running multiple parallel \ac{SCL} decodings on permuted received words, yielding a list of estimates from which the final estimate is selected.
Yet, \ac{AED} is limited to appropriately designed polar codes.
\Ac{SCED} was recently proposed for \acl{LDPC} codes and does not impose such design constraints.
It uses multiple decodings in different subcodes, ensuring that the selected subcodes jointly cover the original code.
We extend \ac{SCED} to polar codes by expressing polar subcodes through suitable \acp{PT}. 
To this end, we describe a framework classifying pre-transformations for pre-transformed polar codes based on their role in encoding and decoding.
Within this framework, we propose a new type of \ac{PT} enabling \ac{SCED} for polar codes, analyze its properties, and discuss how to construct an efficient ensemble. 
\end{abstract}
\acresetall

\section{Introduction}
\blfootnote{This work has received funding from the 
German Federal Ministry of Education and Research (BMBF) within the project Open6GHub (grant agreements 16KISK010 and 16KISK019) and the European Research Council (ERC) under the European Union’s Horizon 2020 research and innovation program (grant agreement No. 101001899).}

\vspace*{-1.2em}
Polar codes are the first family of block codes that provably achieve the capacity of \ac{BDMC} together with low-complexity \ac{SC} decoding asymptotically in the block length \cite{arikan_channel_2009}.
Yet, it is due to their excellent performance at short block lengths---when pre-transformed with a \ac{CRC} code and decoded using \ac{SCL} decoding---that they have been standardized for the control channel in the 5G standard \cite{tal_list_2015,bioglio_design_2020}.
With the anticipated growth of machine-type communication, the requirements on the performance of short block codes are expected to become more stringent in 6G \cite{Miao2024Trends}.
To ensure that polar codes remain a strong contender for the short block length regime, their performance must be further improved, e.g., by enhancing \ac{SCL}-based decoding.
However, practical list sizes are limited, as increasing the list size leads to suboptimal area scaling in hardware implementations.
For instance, \cite[Tab.~1]{10104534}, reports that doubling the list size from $8$ to $16$ increases the required chip area by more than a factor of $5$.
Thus, enhancing \ac{SCL}-based decoding while maintaining a fixed list size is of interest.

Ensemble decoding schemes can improve the error correction performance in the short block length scenario \cite{AED_RMcodes,mandelbaum_subcode_2025,kraft_ensemble_2023,mandelbaum2024endomorphisms,MBBP1}.
While they all rely on parallel paths with independent decoding instances, they can be broadly categorized into two classes \cite{mandelbaum_subcode_2025}: those that alter the noise representation---e.g., \ac{AED} which permutes the received word using an automorphism of the code \cite{AED_RMcodes}---and those that perform decoding on alternative graphical representations \cite{MBBP1}.
The inherent parallel structure of ensemble decoding schemes enables efficient parallel implementation and, in addition, can yield improved error correction performance %
\cite{10104534,mandelbaum_subcode_2025}.
For instance, \ac{SCL}-based \ac{AED} can significantly improve performance compared to stand-alone \ac{SCL} decoding with a larger list size \cite[Fig. 9]{AED_RMcodes}.
However, when considering polar codes, the beneficial effect of automorphisms can be potentially absorbed by symmetries of the decoder \cite[Theorem~2]{AED_RMcodes}.
Hence, \ac{AED} imposes additional constraints, limiting its application to suitably designed polar codes \cite{PolarAEDPillet,Aut_PolarCodes_Geiselhart}. 

\Ac{SCED}, which was introduced for \ac{BP}-based decoding in \cite{mandelbaum_subcode_2025}, uses an ensemble of subcode decodings, where the subcodes jointly cover the original code. \Ac{SCED} does not impose any constraints onto the code design. %
In this work, we consider \ac{SCED} for polar codes by introducing a novel kind of polar \acp{PT} which generates suitable subcodes. 
By optimizing the selection of such \acp{PT}, we show that \ac{SCED} offers competitive performance compared to conventional \ac{SCL} decoding.

\section{Preliminaries}
\subsection{Polar Codes}
Polar codes leverage the $n$-fold \emph{polar transform} to synthesize virtual bit-channels given a set of identical \acp{BDMC} \cite{arikan_channel_2009}. These virtual bit-channels are \emph{polarized}, i.e., either reliable---suitable for uncoded information transmission---or unreliable with their input set to a known frozen value (typically $0$).
The indices of the reliable and unreliable bit-channels form the information set~$\mathcal{I}$ and the frozen set~$\mathcal{F}$, respectively. 
A polar code $\mathcal{C}(N, k)$ with block length $N=2^n,n\in \mathbb{N}$, is defined by the $n$-fold Hadamard matrix $\bm{G}_N=\begin{psmallmatrix}1 & 0 \\ 1 & 1 \end{psmallmatrix}^{\otimes n}$ along with its information set $\mathcal{I}\subseteq[N]$ of size $|\mathcal{I}|=k$, or equivalently, by its frozen set $\mathcal{F}=[N]\setminus\mathcal{I}$. Here, $(\cdot)^ {\otimes n}$ denotes the $n$-fold application of the Kronecker product and $[N]:=\{0,\ldots,N-1\}$.
The design of a polar code consists of carefully identifying the most reliable bit-channels, i.e., determining $\mathcal{I}$. 
To this end, a widely used approach for the \ac{BI-AWGN} channel is \ac{DE} \cite{mori_performance_2009}. 
For the application of polar codes in 5G, a reliability sequence is defined in \cite{bioglio_design_2020}.

Encoding begins by mapping the data word ${\bm{u}\in\mathbb{F}_2^k}$ onto the padded data word $\bm{u}_\mathrm{p}\in\mathbb{F}_2^N$, where $(\bm{u}_\mathrm{p})_\mathcal{I}=\bm{u}$ and  ${(\bm{u}_\mathrm{p})_\mathcal{F}=\bm{0}}$.
Hereby, $(\bm{u})_\mathcal{A}$ denotes the vector consisting of the elements from $\bm{u}$ with indices in $\mathcal{A}$. 
Then,  multiplication with $\bm{G}_N$ yields the codeword $\bm{x}\in\mathcal{C}(N, k)$, i.e, $\bm{x}=\bm{u}_\mathrm{p}\bm{G}_N$. 
Polar codes were originally proposed together with \ac{SC} decoding, which follows a sequential decoding order \cite{arikan_channel_2009}, decoding bits with lower indices first \cite{alamdar-yazdi_simplified_2011}. This order prevents \ac{SC} decoding from revising early decisions. 
To overcome this limitation, \ac{SCL} decoding allows the decoder to maintain a list of possible candidate sequences
and can revert to decisions that turn out to be more likely as decoding progresses~\cite{tal_list_2015}.

\subsection{Pre-transformed Polar Codes}

In \cite{tal_list_2015}, the authors observe that often, when \ac{SCL} decoding produces an incorrect codeword estimate, the correct estimate is in the list of candidates, but is discarded in favor of a more likely codeword.
Hence, they propose to pre-transform the data word $\bm{u}$ with a \ac{CRC} code before polar encoding and use the \ac{CRC} code as a decision genie when choosing the final estimate.
This significantly improves the error-correcting performance and reduces the undetected error rate compared to conventional \ac{SCL} decoding.

Other pre-transformed polar codes modify the value of frozen bits, e.g., polarization-adjusted convolutional \cite{arıkan2019pac} and row-merged polar codes \cite{rowmerged_zunker_2025}. Instead of selecting $0$ as the value of frozen bits, the value depends on information bits with lower indices and, thus, the frozen bits become \emph{dynamic} frozen bits.
Under \ac{SCL} decoding, this approach offers improved error correction performance due to the improved weight spectrum and the exploitation of the remaining capacity of the frozen bit-channels \cite{rowmerged_zunker_2025}.

\subsection{Subcodes and Subcode Ensemble Decoding (ScED)}
We define a subcode $\mathcal{C}_{T}(N,k)\subseteq \mathcal{C}(N,\kappa)$ of a polar code $\mathcal{C}(N,\kappa)$ to 
consist of $2^k$ codewords out of the $2^{\kappa}$ polar codewords, where $k\leq \kappa$ and $T$ is a suitable pre-transformation.
In contrast to \cite{trifonov_subcodes_2016,mandelbaum_subcode_2025}, we do not require that $\mathcal{C}_{T}(N,k)$ is a linear subspace, and we generalize the definition of \mbox{pre-transformation} in \cite{trifonov_subcodes_2016} by allowing for an offset.

\ac{SCED}, introduced in \cite{mandelbaum_subcode_2025}, consists of $M$ parallel paths and a \ac{ML}-in-the-list decision.
Assume the transmission of an arbitrary codeword $\bm{x}$ of a linear block code $\mathcal{C}(N,\kappa)$ over a channel with channel output alphabet $\mathcal{Y}$ where $\bm{y}\in\mathcal{Y}^N$ is observed at the channel output.
Then, each path of \ac{SCED} performs decoding on a subcode $\mathcal{C}_i\subseteq\mathcal{C}$, where $i\in[M]$, yielding $M$ possibly different estimates $\hat{\bm{x}}_i$.
If $\hat{\bm{x}}_i\in \mathcal{C}_i$, the $i$th subcode decoding converged successfully, which is not guaranteed for
\ac{BP} decoding in \cite{mandelbaum_subcode_2025}.
Given the set $\mathcal{L}=\{\hat{\bm{x}}_i\}_{i\in[M]}$ containing the estimates of the $M$ different paths, the \ac{ML}-in-the-list decision selects the estimate that belongs to the code $\mathcal{C}$ and that maximizes the log-likelihood function $L(\bm{y}|\bm{x})$.
If no subcode decoding yields a valid codeword, the estimate maximizing $L(\bm{y}|\bm{x})$ is chosen to minimize the \ac{BER}.
There might exist codewords $\bm{x}\in \mathcal{C}$ with $\bm{x}\not\in \mathcal{C}_i\,\forall i\in[M]$. To avoid this, \ac{SCED} employs a carefully chosen ensemble of $M$ subcodes such that
\begin{equation}\label{eq:union_of_subcodes}
    \bigcup_{i\in[M]} \mathcal{C}_i = \mathcal{C},
\end{equation}
ensuring that all codewords are contained in at least one subcode \cite{mandelbaum_subcode_2025}.

\section{Subcode Ensemble Decoding of Polar Codes}   
In this section, we introduce \ac{SCED} for polar codes. Linear subcodes of polar codes can be represented using suitable \acp{PT}\cite{trifonov_subcodes_2016}.
To provide a unified framework, we classify three types of \acp{PT}, as shown in Tab.~\ref{tab:classified_pretransformations}, based on their use in encoding and decoding. 
First, in Sec.~\ref{sec:generalized_framework}, we describe \PTA~and \PTB~which incorporate well-known classes of pre-transformed polar codes.
Then, in Sec.~\ref{sec:type_c}, we introduce a novel third type of \acp{PT}, denoted as \PTC, which are considered solely during decoding to enable \ac{SCED}.

\begin{table}[t]
    \centering
        \caption{Classification of the three types of pre-transformations.}
    \begin{tabular}{cccc}
        \toprule
            Type~of PT & Encoding & Decoding & Decision-Genie \\
            \midrule
         \textcolor{kit-red100}{\textbf{PT-A}} & \checkmark & $\times$ & \checkmark \\ 
         \textcolor{KITpurple}{\textbf{PT-B}} & \checkmark & \checkmark & $\times$  \\
         \textcolor{kit-blue100}{\textbf{PT-C}} & $\times$ & \checkmark & $\times$ \\ 
         \bottomrule
        \end{tabular}
    \label{tab:classified_pretransformations}
\end{table}

\subsection{Generalized Framework for Pre-transformations}
\label{sec:generalized_framework}
In our framework, we consider \acp{PT} which are affine transformations $T:\mathbb{F}_2^N\rightarrow\mathbb{F}_2^N$ mapping the padded data word ${\bm{u}_\mathrm{p} \in \mathbb{F}_2^N}$ onto another padded data word ${\tilde{\bm{u}}_\mathrm{p} \in \mathbb{F}_2^N}$: %
\begin{equation}
            T(\bm{u}_\mathrm{p}) = \bm{u}_\mathrm{p}\bm{A} + \bm{b} = \tilde{\bm{u}}_\mathrm{p}, \qquad \bm{A} \in \mathbb{F}_2^{N\times N},\quad \bm{b} \in \mathbb{F}_2^{N}.
            \label{eq:pt}
\end{equation}
Hereby, following \cite{rowmerged_zunker_2025,trifonov_subcodes_2016}, the matrix $\bm{A}$ is upper triangular to ensure that dynamic frozen bits are only dependent on previously decoded information bits.
In contrast to \cite{trifonov_subcodes_2016}, (\ref{eq:pt}) allows for an offset $\bm{b}$ such that the resulting subcodes are potentially affine subspaces of the polar code $\mathcal{C}(N, \kappa)$. 
As illustrated in Fig.~\ref{fig:entire_graph_encoding}, applying a \ac{PT} to a polar code results in a \emph{pre-transformed polar code} $\mathcal{C}_T(N, k)$ with $k \leq \kappa$. 
Furthermore, we denote the graph that represents $\mathcal{C}_T(N, k)$ and which consists of the graphical representations of the \ac{PT} and the polar code as the \emph{joint graph}.

Fig.~\ref{fig:jointgraph} depicts an exemplary joint graph of a pre-transformed polar code $\mathcal{C}_T(8,3)$ where the information bits and frozen bits of the polar code $\mathcal{C}(8,4)$ are colored in \textcolor{KITorange}{orange} and black, respectively.
We decompose the bit indices in $\tilde{\bm{u}}_\mathrm{p}$ as $[N]=\mathcal{T}\cup\mathcal{T}^c$.
Herein, $\mathcal{T}^c$ denotes the indices of bits that do not undergo pre-transformation, i.e., $(\tilde{\bm{u}}_\mathrm{p})_i=(\bm{u}_\mathrm{p})_i, i\in\mathcal{T}^c,$ and are depicted as circles 
(\blackcirc\,/\,\orangecirc)
in Fig.~\ref{fig:jointgraph}, whereas indices in $\mathcal{T}$ that are subject to an affine mapping are called \emph{target bits} and are depicted as triangles
(\blacktri\,/\,\orangetri). 
The bits in $\bm{u}_\mathrm{p}$ that contribute to a certain target bit are denoted \emph{origin bits}. 
We assume that target bits do not depend on an origin bit with the same index and, thus, $A_{t,t}=0, \, \forall t \in \mathcal{T}$. 
The number of target bits for a given \ac{PT} is referred to as the \emph{depth} $d_\mathrm{p}=|\mathcal{T}|$ of the \ac{PT}.
Note that $k=\kappa-d_\mathrm{p}$. 

\begin{figure}[!t]
    \centering
    \includegraphics[scale=1]{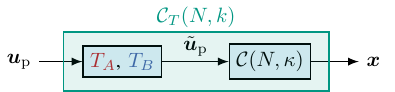}
    \caption{Pre-transformed polar code formed by employing \PTA~and \PTB.}
    \label{fig:entire_graph_encoding}
\end{figure}

\begin{figure}[t]
    \centering
    \includegraphics[scale=1]{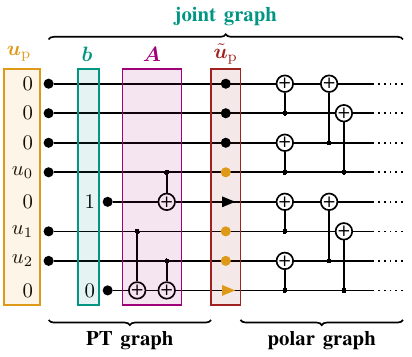}
    \caption{Exemplary joint encoding graph for the polar code $\mathcal{C}(8, 4)$ with information set ${\mathcal{I} = \{3, 5, 6, 7\}}$ and \ac{PT} with $\mathcal{T}=\{4, 7\}$ from (\ref{eq:PTexample}), resulting in the pre-transformed polar code $\mathcal{C}_T(8,3)$. Information bits of $\mathcal{C}$ are colored in \textcolor{KITorange}{orange}. Target bits $t \in \mathcal{T}$ of the \ac{PT} are drawn as triangles (\protect\blacktri\,/\,\protect\orangetri).}
    \label{fig:jointgraph}
\end{figure}

The first kind of \acp{PT}, \PTA~with target bits $\mathcal{T}_\mathrm{A}$, are considered in encoding, while their target bits with $\mathcal{T}_\mathrm{A} \subseteq \mathcal{I}$ are treated as conventional information bits upon decoding. 
Thus, the decoding operates on $\mathcal{C}$, which is a supercode of $\mathcal{C}_T$, i.e., $\mathcal{C}_T \subseteq \mathcal{C}$. 
This effect is leveraged in genie-aided decisions to filter out invalid codewords from the list in order to improve decoding performance and lower the undetected error rate. 
For instance, \ac{CRC}-aided polar codes can be represented in our framework by a polar code with a \PTA. 

In contrast, \PTB, which incorporate row-merged polar codes \cite{rowmerged_zunker_2025}, are considered in both encoding and decoding. 
To decode on the joint graph, the decoder considers the \ac{PT} as the decoding process reaches the target bits \cite{rowmerged_zunker_2025,arıkan2019pac}.
Thus, the target bits have to be included in the frozen set, i.e., $\mathcal{T}_\mathrm{B}\subseteq\mathcal{F}$ and become \emph{dynamic} frozen by applying the \ac{PT}.

For illustration, Fig.~\ref{fig:jointgraph} depicts an exemplary joint graph of a polar code $\mathcal{C}(8, 4)$ with ${\mathcal{I} = \{3, 5, 6, 7\}}$ together with a \ac{PT} with $b_4=1$ and the other offsets being $0$, resulting in
\begin{align}
\begin{split}
    \tilde{\bm{u}}_\mathrm{p} &= 
    \left(\begin{array}{cccccccc}
         0 & 0 & 0 & u_0 & u_0 + 1 & u_1 & u_2 & u_1 + u_2
    \end{array} \right).
    \label{eq:PTexample}
    \end{split}
\end{align}

Here, $\tilde{u}_{\mathrm{p},4}$ is a target bit of a \PTB~(\blacktri) as $4\notin\mathcal{I}$, whereas the \ac{PT} modifying $\tilde{u}_{\mathrm{p},7}$ is a \PTA~(\orangetri) since $7\in \mathcal{I}$.

\subsection{Applying Pre-transformations in ScED}
\label{sec:type_c}
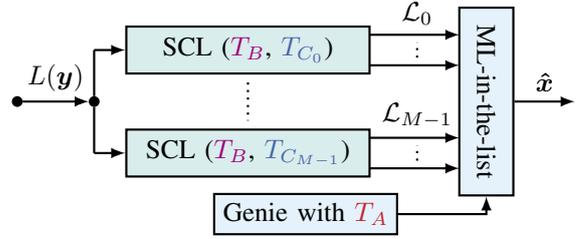
\begin{figure}[t]
    \centering
    \resizebox{.88\columnwidth}{!}{
    \begin{tikzpicture}[thick]
    \tikzstyle{rect} = [rectangle, fill=kit-green15, draw=black, minimum height = 17.5pt, minimum width = 90pt]
    \tikzstyle{smallcircle} = [circle, fill=black, minimum size=4pt, inner sep=0pt, line width=0]
    \tikzstyle{decrect} = [rectangle, draw, minimum height=20pt, rotate=-90, fill=kit-cyanblue15]
    \tikzstyle{channel} = [rectangle, draw, fill=kit-cyanblue15]

    \newcommand\yspacing{1.35}
    \newcommand\gapfac{2.5}
    \newcommand\mlxpos{2.75}

    \node[smallcircle] (inp) at (-3, -\yspacing / 2) {};
    
    \coordinate (d0+1) at (\mlxpos, 0);
    \coordinate (d2+1) at (\mlxpos, -\yspacing);

    \coordinate (d0-1) at (-2, 0);
    \coordinate (d2-1) at (-2, -\yspacing);

    \node[smallcircle] (d-1) at (-2, -\yspacing / 2) {};

    \node[decrect, text centered, minimum width=70pt] (decide) at (\mlxpos + 0.35, -\yspacing / 2) {ML-in-the-list};

    \node[rectangle, fill=kit-cyanblue15, draw] (genie) at (\mlxpos - 2, -\yspacing - 0.8) {Genie with \textcolor{kit-red100}{$T_A$}};

    \coordinate (out) at (\mlxpos + 1.5, -\yspacing / 2);

    \coordinate (d0start1) at (0, 0.2);
    \coordinate (d0start2) at (0, -0.2);
    \coordinate (d0+1end1) at (\mlxpos, 0.2);
    \coordinate (d0+1end2) at (\mlxpos, -0.2);
    
    \coordinate (d2start1) at (0, 0.2 - \yspacing);
    \coordinate (d2start2) at (0, -0.2 - \yspacing);
    \coordinate (d2+1end1) at (\mlxpos, 0.2 - \yspacing);
    \coordinate (d2+1end2) at (\mlxpos, -0.2 - \yspacing);
    
    \draw[-latex] (d0start1) -- (d0+1end1) node[pos=0.8, above] {$\mathcal{L}_0$};
    \draw[-latex] (d0start2) -- (d0+1end2);

    \draw[-latex] (d2start1) -- (d2+1end1) node[pos=0.8, above] {$\mathcal{L}_{M-1}$};
    \draw[-latex] (d2start2) -- (d2+1end2);   

    \draw[dotted] (\mlxpos - 0.55, -0.1) -- (\mlxpos - 0.55, 0.1);
    \draw[dotted] (\mlxpos - 0.55, -0.1 - \yspacing) -- (\mlxpos - 0.55, 0.1 - \yspacing);

    \node[rect] (d0) at (0, 0) {SCL (\textcolor{kit-purple100}{$T_B$}, \textcolor{KITblue}{$T_{C_0}$})};
    \node[rect] (d2) at (0, -\yspacing) {SCL (\textcolor{kit-purple100}{$T_B$}, \textcolor{KITblue}{$T_{C_{M-1}}$})};

    \draw[-latex] (inp) -- (d-1) node[pos=0.5, above] {$L(\bm{y})$};

    \draw[-latex] (decide) -- (out) node[midway, above] {$\bm{\hat{x}}$};

    \draw[-latex] (genie) -- (\mlxpos + 0.35, -\yspacing - 0.8) -- (decide);

    \draw[-latex] (d0-1) -- (d0);
    \draw[-latex] (d2-1) -- (d2);

    \draw[-] (d-1) -- (d0-1);
    \draw[-] (d-1) -- (d2-1);

    \coordinate (help00) at (0, -\yspacing / 2 + 0.25);
    \coordinate (help01) at (0, -\yspacing / 2 - 0.25);
    \draw[dotted] (help00) -- (help01);

\end{tikzpicture}
    }
    \caption{\ac{SCED} employing \PTB~and $M$ \PTC s. \PTA s~are considered for the genie-aided decision.}
    \label{fig:entire_graph_decoding}
\end{figure}

Let $\mathcal{C}(N,\kappa)$ denote the polar code that is used for transmission. 
The goal of a \PTC~is to constitute a joint graph that is used in decoding but \emph{not} used for encoding. 
Thus, while trying to decode a codeword $\bm{x} \in \mathcal{C}$, the decoder operates on the joint graph associated with a subcode $\mathcal{C}_T(N, k)$ of $\mathcal{C}(N,\kappa)$ enabling \ac{SCED} for polar codes.
To this end, similar to \PTB, target bits have indices ${\mathcal{T}_C\subseteq\mathcal{I}}$ and, consequently, become dynamic frozen bits by decoding on the joint graph.
Note that if $\bm{x} \in \mathcal{C}$ and $\bm{x} \notin \mathcal{C}_T(N, k)$, the subcode decoding, i.e., the decoding on the joint graph, cannot yield the correct estimate. 
Therefore, following \cite{mandelbaum_subcode_2025}, we select $M$ different \PTC s such that the respective subcodes fulfill~(\ref{eq:union_of_subcodes})~\cite{mandelbaum_subcode_2025}.
 Then, as illustrated in Fig.~\ref{fig:entire_graph_decoding}, we combine these $M$ subcode decodings to constitute an \ac{SCED} for polar codes.

Next, we introduce Theorem~\ref{theorem:samedecision}, where we investigate the decoding behavior of \ac{SC}-based \ac{SCED} compared to stand-alone \ac{SC} decoding. The analysis of \ac{SCL} decoding will be part of future work.
\begin{theorem}
    Consider a polar code $\mathcal{C}(N, k)$ with \ac{SC} decoding $\mathrm{SC}: \mathcal{Y}^N \to \mathcal{C}$. Consider a \PTC~such that $\mathcal{C}_T \subseteq \mathcal{C}$ and let $\mathrm{SC}_T: \mathcal{Y}^N \to \mathcal{C}_T$ denote the decoder on the joint graph. Assume that codeword $\bm{x} \in \mathcal{C}_T$ is transmitted over a channel and $\bm{y} \in \mathcal{Y}^N$ is received at the output. Then,
    \begin{equation}\label{eq:sced_outperforms_sc}
        \mathrm{SC}(\bm{y})=\bm{x} \implies \mathrm{SC}_T(\bm{y}) = \bm{x},
    \end{equation}
    i.e., if a stand-alone \ac{SC} decoding decodes correctly, a subcode \ac{SC} decoding will make the same decision, provided the codeword is included in its corresponding subcode.
    \label{theorem:samedecision}
\end{theorem}

\begin{proof}
    Consider an unmodified \ac{SC} decoder $\mathrm{SC}(\bm{y})$ and an \ac{SC}-based decoding on the joint graph of a depth ${d_\mathrm{p}=1}$ \ac{PT}, denoted as $\mathrm{SC}_T(\bm{y})$, which yields the subcode $\mathcal{C}_T$.
    Assume that $\bm{x} \in \mathcal{C}_T$ is transmitted over a channel and $\bm{y} \in \mathcal{Y}^N$ is received at the output.
    The decoding process of $\mathrm{SC}_T$ and SC is identical until reaching the target bit of the \ac{PT}.
    Assuming $\mathrm{SC}(\bm{y})=\bm{x}$ and because $\bm{x}\in \mathcal{C}_T$, both decodings will make the same decision at that bit.    
    As subsequent decisions solely depend on previous decisions and the received values $\bm{y}$, both decodings will continue to make the same decisions and, thus, will return the same estimate, i.e., $\mathrm{SC}_T(\bm{y}) = \mathrm{SC}(\bm{y}) = \bm{x}$.
    By considering each target bit independently, the same argument holds for an arbitrary \PTC~with depth $d_\mathrm{p} > 1$ if $\bm{x}\in \mathcal{C}_T$.
\end{proof}
Note that ${\mathrm{SC}_T(\bm{y})=\bm{x} \notimplies \mathrm{SC}(\bm{y}) = \bm{x}}$. Hence, together with Theorem~\ref{theorem:samedecision}, this motivates that the average performance of \ac{SC}-based \ac{SCED} is at least as good as the performance of \ac{SC} decoding when ensuring a covering of the original code $\mathcal{C}$ as proposed in \cite{mandelbaum_subcode_2025}.
Furthermore, the performance cannot be improved by incorporating a path that performs \ac{SC} decoding on the original polar graph which is in contrast to \ac{BP}-based \ac{SCED}, where the ensemble decoding still benefits from a path that employs decoding on the original graph \cite{mandelbaum_subcode_2025}.

\section{Selection of Pre-transformations for ScED}
The selection of suitable \PTC s for \ac{SCED} plays a crucial role in its decoding performance.
In this section, we outline two steps to identify ensembles with promising error correction capabilities.
First, we analyze the influence of the parameters of \PTC s on the performance of the joint graph decoding. 
Using those insights, we generate a large set of candidate subcode decodings from which we select a small subset of
\PTC~to form an \ac{SCED}.
To achieve this, we simulate the transmission of codewords over a 
\ac{BI-AWGN} channel collecting \acl{LLR} patterns that result in a frame error when decoded using stand-alone \ac{SC} decoding or SCL-$L$ decoding. We refer to these patterns as \acp{URP}-SC or \acp{URP}-SCL-$L$, respectively.

\subsection{Parameter Analysis}
We consider a $\mathcal{C}(256, 128)$ polar code with an information set designed using \ac{DE} and repeat the following process $2000$ times: we sample $100$ \acp{URP}-SC at $E_\mathrm{b}/N_0=2.5\,\mathrm{dB}$ and select $192$ random \PTC s.
Then, for each \PTC, we track the depth and target bits and evaluate the number of \acp{URP}-SC that the respective subcode can decode correctly.

In Fig.~\ref{fig:performance_analysis_target_bits} we depict the ratio of decoded \acp{URP} for different target bit indices. 
Hereby, to assess the influence of each potential target bit individually, we solely consider \acp{PT} of depth $d_\mathrm{p}=1$.  For each target bit, we sample origin bits randomly from the bits that are decoded earlier in \ac{SC} decoding and offsets according to $\mathrm{Bernoulli}\mleft(\frac{1}{2}\mright)$.
We sort the target bits according to the reliability of the corresponding bit-channels in ascending order.
Target bits associated with unreliable bit-channels typically yield a higher ratio of decoded \acp{URP} than those corresponding to more reliable bit-channels.
This effect arises because target bits of \PTC s become dynamic frozen bits and thus no longer negatively impact decoding performance.
However, some outliers indicate that additional effects influence the decoding performance of subcode decodings.

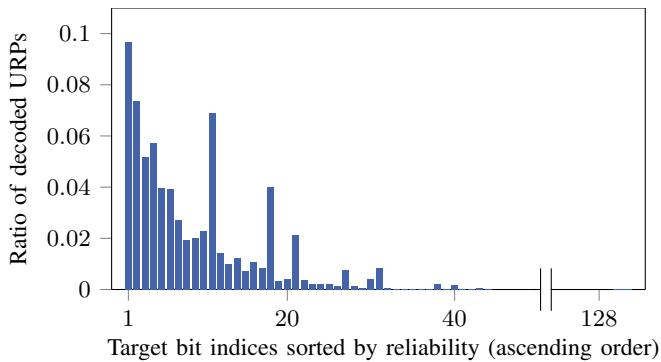
\begin{figure}[t]
    \centering
    \pgfplotsset{
    every non boxed y axis/.style={} 
}

\begin{tikzpicture}

\definecolor{darkgray176}{RGB}{176,176,176}

\begin{groupplot}[
    group style={
        group name=my fancy plots,
        group size=2 by 1,
        yticklabels at=edge left,
        horizontal sep=0pt
    },
    height = \columnwidth * 0.6,
    ymin=0, ymax=0.11,
    ytick={0, 0.02, 0.04, 0.06, 0.08, 0.1},
    xticklabel style = {font=\small},
    yticklabel style = {font=\small},   
    separate axis lines,
    axis x line*=bottom
]

\nextgroupplot[xmin=-1, xmax=49,
               xtick={1, 20, 40},
               axis y line=left, 
               ylabel={\small{Ratio of decoded URPs}},
               width=\columnwidth*0.8,
               scaled ticks=false,
               yticklabel style={/pgf/number format/fixed},
                tick align=outside,
                tick pos=left,
                xlabel={\small{Target bit indices sorted by reliability (ascending order)}},
                x label style={at={(axis description cs:0.65,0)},anchor=north},]
\draw[draw=none,fill=KITblue] (axis cs:0.6,0) rectangle (axis cs:1.4,0.0965549597855228);
\draw[draw=none,fill=KITblue] (axis cs:1.6,0) rectangle (axis cs:2.4,0.0735161507402423);
\draw[draw=none,fill=KITblue] (axis cs:2.6,0) rectangle (axis cs:3.4,0.051745867768595);
\draw[draw=none,fill=KITblue] (axis cs:3.6,0) rectangle (axis cs:4.4,0.0572246401071309);
\draw[draw=none,fill=KITblue] (axis cs:4.6,0) rectangle (axis cs:5.4,0.0397513623978202);
\draw[draw=none,fill=KITblue] (axis cs:5.6,0) rectangle (axis cs:6.4,0.0391232148787778);
\draw[draw=none,fill=KITblue] (axis cs:6.6,0) rectangle (axis cs:7.4,0.0273361870266711);
\draw[draw=none,fill=KITblue] (axis cs:7.6,0) rectangle (axis cs:8.4,0.0194613501224204);
\draw[draw=none,fill=KITblue] (axis cs:8.6,0) rectangle (axis cs:9.4,0.0199373259052925);
\draw[draw=none,fill=KITblue] (axis cs:9.6,0) rectangle (axis cs:10.4,0.0227885906040268);
\draw[draw=none,fill=KITblue] (axis cs:10.6,0) rectangle (axis cs:11.4,0.0691523678792038);
\draw[draw=none,fill=KITblue] (axis cs:11.6,0) rectangle (axis cs:12.4,0.0143072882468812);
\draw[draw=none,fill=KITblue] (axis cs:12.6,0) rectangle (axis cs:13.4,0.00988239247311828);
\draw[draw=none,fill=KITblue] (axis cs:13.6,0) rectangle (axis cs:14.4,0.0124110671936759);
\draw[draw=none,fill=KITblue] (axis cs:14.6,0) rectangle (axis cs:15.4,0.0074169741697417);
\draw[draw=none,fill=KITblue] (axis cs:15.6,0) rectangle (axis cs:16.4,0.0107913425769361);
\draw[draw=none,fill=KITblue] (axis cs:16.6,0) rectangle (axis cs:17.4,0.00831924577373212);
\draw[draw=none,fill=KITblue] (axis cs:17.6,0) rectangle (axis cs:18.4,0.0399545602077248);
\draw[draw=none,fill=KITblue] (axis cs:18.6,0) rectangle (axis cs:19.4,0.00317316994517897);
\draw[draw=none,fill=KITblue] (axis cs:19.6,0) rectangle (axis cs:20.4,0.00411107348459194);
\draw[draw=none,fill=KITblue] (axis cs:20.6,0) rectangle (axis cs:21.4,0.0211689620423245);
\draw[draw=none,fill=KITblue] (axis cs:21.6,0) rectangle (axis cs:22.4,0.00383593490534706);
\draw[draw=none,fill=KITblue] (axis cs:22.6,0) rectangle (axis cs:23.4,0.00223446893787575);
\draw[draw=none,fill=KITblue] (axis cs:23.6,0) rectangle (axis cs:24.4,0.00222473361317404);
\draw[draw=none,fill=KITblue] (axis cs:24.6,0) rectangle (axis cs:25.4,0.00205153949129853);
\draw[draw=none,fill=KITblue] (axis cs:25.6,0) rectangle (axis cs:26.4,0.00138259441707718);
\draw[draw=none,fill=KITblue] (axis cs:26.6,0) rectangle (axis cs:27.4,0.00775787728026534);
\draw[draw=none,fill=KITblue] (axis cs:27.6,0) rectangle (axis cs:28.4,0.00126417038818276);
\draw[draw=none,fill=KITblue] (axis cs:28.6,0) rectangle (axis cs:29.4,0.00054637436762226);
\draw[draw=none,fill=KITblue] (axis cs:29.6,0) rectangle (axis cs:30.4,0.00408598510494245);
\draw[draw=none,fill=KITblue] (axis cs:30.6,0) rectangle (axis cs:31.4,0.00821332436069987);
\draw[draw=none,fill=KITblue] (axis cs:31.6,0) rectangle (axis cs:32.4,0.000554106910039113);
\draw[draw=none,fill=KITblue] (axis cs:32.6,0) rectangle (axis cs:33.4,0.000330771800868694);
\draw[draw=none,fill=KITblue] (axis cs:33.6,0) rectangle (axis cs:34.4,0.000240131578947368);
\draw[draw=none,fill=KITblue] (axis cs:34.6,0) rectangle (axis cs:35.4,0.00034248788368336);
\draw[draw=none,fill=KITblue] (axis cs:35.6,0) rectangle (axis cs:36.4,8.70125684821141e-05);
\draw[draw=none,fill=KITblue] (axis cs:36.6,0) rectangle (axis cs:37.4,3.30797221303341e-06);
\draw[draw=none,fill=KITblue] (axis cs:37.6,0) rectangle (axis cs:38.4,0.00233857477417196);
\draw[draw=none,fill=KITblue] (axis cs:38.6,0) rectangle (axis cs:39.4,4.96688741721854e-05);
\draw[draw=none,fill=KITblue] (axis cs:39.6,0) rectangle (axis cs:40.4,0.00183759301197024);
\draw[draw=none,fill=KITblue] (axis cs:40.6,0) rectangle (axis cs:41.4,0);
\draw[draw=none,fill=KITblue] (axis cs:41.6,0) rectangle (axis cs:42.4,1.69606512890095e-05);
\draw[draw=none,fill=KITblue] (axis cs:42.6,0) rectangle (axis cs:43.4,0.000777853260869565);
\draw[draw=none,fill=KITblue] (axis cs:43.6,0) rectangle (axis cs:44.4,0);

\path[-] (rel axis cs:0,1)     coordinate(topstart)
         --(rel axis cs:1,1)   coordinate(topstop);
\draw(topstart) -- (topstop);

\nextgroupplot[xmin=125, xmax=130,
               xtick={128},
               ytick={},
               xtick pos=bottom,
               ytick style={color=white},
               tick align=outside,
               xticklabel style = {font=\small},
               yticklabel style = {font=\small},
               axis y line=right, 
               axis x discontinuity=parallel,
               width=\columnwidth*0.35]

\draw[draw=none,fill=KITblue] (axis cs:128.6,0) rectangle (axis cs:129.4,0);
\path[-] (rel axis cs:0,1)     coordinate(topstart)
         --(rel axis cs:1,1)   coordinate(topstop);
\draw(topstart) -- (topstop);
    
\end{groupplot}

\end{tikzpicture}
    \vspace*{-1.8em}
    \caption{Ratio of decoded \acp{URP} over the indices of target bits of \acp{PT} sorted by the reliability of the corresponding bit-channels for the \ac{DE} polar code $\mathcal{C}(256,128)$ with SC decoding in ascending order. 
    Analysis based on ${2000 \cdot 100}$ \acp{URP} using $192$ newly sampled subcode decoders per $100$ \acp{URP}.}
    \label{fig:performance_analysis_target_bits}
\end{figure}

To investigate the impact of the depth $d_\mathrm{p}$ of a \PTC, we 
extend our previous analysis and
sample $d_\mathrm{p}$ uniformly from $\{1,\dots,15\}$. 
On the one hand, a larger depth increases the number of dynamic frozen bits, potentially improving performance. 
On the other hand, the number of codewords of the polar code that are elements of the subcode decreases exponentially with increasing $d_\mathrm{p}$, as $|\mathcal{C}_T|=2^{k-d_\mathrm{p}}$.
Hence, larger depths require more paths in order to possibly fulfill (\ref{eq:union_of_subcodes}) resulting in a trade-off between depth and ensemble size.
This is supported by the results in Fig.~\ref{fig:performance_analysis_depth_de_sc}, which depicts the ratio of decoded \acp{URP} over the depth $d_\mathrm{p}$ of \PTC. For $d_\mathrm{p}>2$, the ratio of decoded \acp{URP} decreases with increasing depth.
However, Fig.~\ref{fig:performance_analysis_depth_de_sc} shows that \acp{PT} with depth $d_\mathrm{p}=2$ tend to outperform those with depth $d_\mathrm{p}=1$, despite the fact that their respective subcodes cover only half as many of the original codewords.
This suggests that setting $d_\mathrm{p}=2$ should yield an efficient ensemble for \ac{SCED}.

\begin{figure}[t]
    \centering
    \begin{tikzpicture}

\definecolor{darkgray176}{RGB}{176,176,176}

\begin{axis}[
width=\columnwidth,
height=\columnwidth * 0.6,
tick align=outside,
tick pos=left,
x grid style={darkgray176},
xlabel={\small{Depth $d_\mathrm{p}$ of pretransformation}},
xmin=-0.5, xmax=15.34,
xtick style={color=black},
y grid style={darkgray176},
xticklabel style = {font=\small},
yticklabel style = {font=\small},
ylabel={\small{Ratio of decoded URPs}},
ytick={0.01, 0.02, 0.03, 0.04},
xtick={1,3,5,7,9,11,13},
ymin=0, ymax=0.0471432279296447,
ytick style={color=black},
scaled ticks=false,
yticklabel style={/pgf/number format/fixed}
]
\draw[draw=none,fill=KITblue] (axis cs:-0.4,0) rectangle (axis cs:0.4,0);
\draw[draw=none,fill=KITblue] (axis cs:0.6,0) rectangle (axis cs:1.4,0.042160878482857);
\draw[draw=none,fill=KITblue] (axis cs:1.6,0) rectangle (axis cs:2.4,0.0448983123139473);
\draw[draw=none,fill=KITblue] (axis cs:2.6,0) rectangle (axis cs:3.4,0.0316948664057566);
\draw[draw=none,fill=KITblue] (axis cs:3.6,0) rectangle (axis cs:4.4,0.0217675461267423);
\draw[draw=none,fill=KITblue] (axis cs:4.6,0) rectangle (axis cs:5.4,0.0142478483738677);
\draw[draw=none,fill=KITblue] (axis cs:5.6,0) rectangle (axis cs:6.4,0.00872900587102697);
\draw[draw=none,fill=KITblue] (axis cs:6.6,0) rectangle (axis cs:7.4,0.00548586137234567);
\draw[draw=none,fill=KITblue] (axis cs:7.6,0) rectangle (axis cs:8.4,0.00301227658521586);
\draw[draw=none,fill=KITblue] (axis cs:8.6,0) rectangle (axis cs:9.4,0.00175899362640342);
\draw[draw=none,fill=KITblue] (axis cs:9.6,0) rectangle (axis cs:10.4,0.0011323521469972);
\draw[draw=none,fill=KITblue] (axis cs:10.6,0) rectangle (axis cs:11.4,0.000593660348911154);
\draw[draw=none,fill=KITblue] (axis cs:11.6,0) rectangle (axis cs:12.4,0.000230867913465449);
\draw[draw=none,fill=KITblue] (axis cs:12.6,0) rectangle (axis cs:13.4,0.000142918232145278);
\draw[draw=none,fill=KITblue] (axis cs:13.6,0) rectangle (axis cs:14.4,0.000076893072805363);

\end{axis}

\end{tikzpicture}
    \vspace*{-1.8em}
    \caption{Ratio of decoded \acp{URP} over the depth $d_\mathrm{p}$ of generated \acp{PT}. Analyzed $2000 \cdot 100$ \acp{URP} using $190$ new decoders referencing \acp{PT} per $100$ \acp{URP} for the \ac{DE} polar code $\mathcal{C}(256,128)$ with SC decoding.}
    \label{fig:performance_analysis_depth_de_sc}
\end{figure}
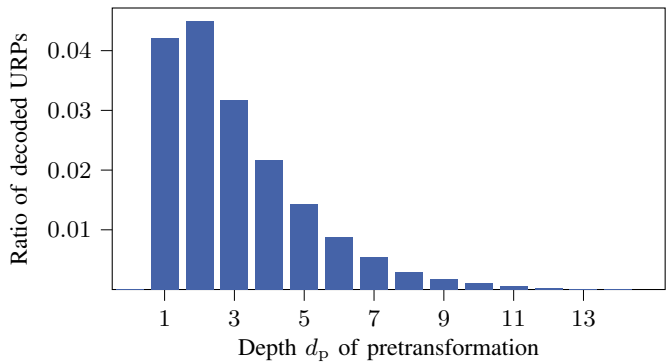

\subsection{Design of an Effective Ensemble for ScED}
\label{sec:design_ensemble}

To design a well-performing \ac{SCED} consisting of $M$ paths of \ac{SCL}-based subcode decodings, we identify a small subset of $M$ candidates from a larger set of $r$ candidates, where $M\ll r$.
To this end, we employ an algorithm that refines the methods proposed in \cite{kraft_ensemble_2023,mandelbaum_subcode_2025}: 
choose $r \in \mathbb{N}$ \PTC s with their associated subcode decoding and sample $s$ \mbox{\acp{URP}-SCL-$L$}.
Then, following \cite{mandelbaum_subcode_2025}, for the $i$th \PTC, we determine the set $\mathcal{S}_i \subseteq [s]$, $i \in [r]$, with $j \in \mathcal{S}_i$ if the $i$th subcode decoding successfully decodes the $j$th \ac{URP}.
In addition to~\cite{mandelbaum_subcode_2025}, we impose the condition that the final estimate of subcode decoding has a lower path metric \cite[Eq. (10)]{balatsoukas-stimming_llr-based_2015} compared to the estimate of stand-alone decoding.
According to \cite{balatsoukas-stimming_llr-based_2015}, the path metric decreases with the likelihood of the codeword.
By definition, stand-alone decoding cannot correctly decode a \ac{URP}.
If stand-alone decoding results in an estimate with a lower path metric compared to the estimate of the subcode decoding,
 the subcode estimate would not prevail in the ML-in-the-list-decision.
Finally, we use the heuristic proposed in \cite{kraft_ensemble_2023} for identifying a subset $\mathcal{E} \subset \{\mathcal{S}_i\}_{i\in[r]}$ of cardinality $|\mathcal{E}| = M$ that maximizes $|\bigcup_{\mathcal{S}_i \in \mathcal{E}}\mathcal{S}_i|$,  i.e., we aim at determining $M$ subcode decodings that  can jointly correct as many \acp{URP} as possible.

\section{Results}
In the following, we consider a very short block length 5G polar code $\mathcal{C}_T(64,32)$ and longer polar codes $\mathcal{C}_T(256,k)$, where $k\in\{64,128,192\}$.
We employ the CRC-$6$ polynomial $\mathrm{0x03}$ and the CRC-$11$ polynomial $0$x$621$ as \PTA, respectively.
To assess the performance of \ac{SCED}, we consider a target \ac{FER} of $10^{-3}$ and conduct Monte-Carlo simulations using a 
\ac{BI-AWGN} channel with at least $1000$ frame errors per simulated point.
The notation \ac{SCED}-$M$-\ac{SCL}-$L$ represents an \ac{SCED} consisting of $M$ \mbox{\ac{SCL}-$L$-based} subcode decodings. The selection of the $M$ \PTC s follows the method proposed in Sec.~\ref{sec:design_ensemble} with ${r=30000}$ and ${s=1000}$.
\mbox{Figs.~\ref{fig:64_32_EDcomparison}--\ref{fig:256_128_FER}} depict the \ac{FER} over $E_\mathrm{b}/N_0$ for \ac{SCED} of the different polar codes, with ensemble size ${M=2}$ and ${M=8}$, respectively.
Furthermore, for comparison, both figures include the performance of \ac{SCL} decoding (using CRC-based list selection) with varying list sizes. For the shorter polar code, we include the performance of \ac{ML} decoding.

At the target \ac{FER}, we observe that \ac{SCED}-$M$-\ac{SCL}-$L$ consistently outperforms \ac{SCL}-$L$ for all considered codes with gains ranging from $0.1\,\mathrm{dB}$ to $0.25\,\mathrm{dB}$.
Furthermore, for the shorter polar code, \ac{SCED} employing two \ac{SCL}-$L$ decodings can match the performance of \ac{SCL}-$2L$ decoding for $L\in\{8, 16\}$, despite the fact that both paths employ only half the list size. 
For the longer codes of rate $\frac{1}{4}$ and $\frac{1}{2}$, respectively, \ac{SCED} requires $8$ rather than $2$ paths to match the decoding performance of stand-alone \ac{SCL}-$2L$ decoding at the target \ac{FER}.

As an \ac{SCL}-$2L$ decoder typically requires more than $5$ times the hardware area of an \ac{SCL}-$L$ decoder \cite{10104534}, increasing the list size is not always feasible. 
\ac{SCED} offers the possibility of efficient hardware re-utilization enabling the decoding performance of \ac{SCL}-$2L$ when chip area is limited.
Alternatively, if area requirements are not as stringent, paths of an \ac{SCED} can be computed in parallel and, thus, reduce decoding latency.

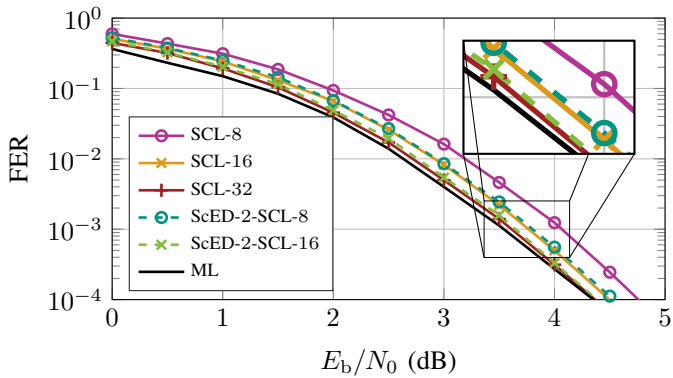
\begin{figure}[t]
    \centering
    \begin{tikzpicture}[spy using outlines={rectangle, magnification=2}]
\begin{axis}[%
width=\columnwidth,
height=\columnwidth * 0.6,
xmin=0,
xmax=5.0,
xlabel style={font=\color{white!15!black}},
xlabel={$E_{\mathrm{b}}/N_0$ \text{(dB)}},
xtick={0, 1.0, 2.0, 3.0, 4.0, 5.0},
ymode=log,
ymin=1e-04,
ymax=1,
ylabel style={font=\color{white!15!black}},
ylabel={FER},
axis background/.style={fill=white},
xmajorgrids,
ymajorgrids,
legend style={at={(0.03,0.03)}, anchor=south west, legend cell align=left, align=left, draw=white!15!black,font=\scriptsize}
]

\addplot[color=KITpurple!80,line width = 1pt, solid,mark=o, mark size=2pt, mark options={solid}]
table[col sep=comma]{
0.0, 5.96e-01
0.5, 4.34e-01
1.0, 3.12e-01
1.5, 1.87e-01
2.0, 9.38e-02
2.5, 4.21e-02
3.0, 1.62e-02
3.5, 4.62e-03
4.0, 1.24e-03
4.5, 2.45e-04
5.0, 4.34e-05
};
\addlegendentry{SCL-$8$};

\addplot[color=KITorange,line width = 1pt, solid, mark=x, mark size=2.5pt, mark options={solid}]
table[col sep=comma]{
0.0, 5.07e-01
0.5, 3.68e-01
1.0, 2.42e-01
1.5, 1.31e-01
2.0, 6.63e-02
2.5, 2.59e-02
3.0, 8.23e-03
3.5, 2.21e-03
4.0, 4.99e-04
4.5, 9.91e-05
};
\addlegendentry{SCL-$16$};

\addplot[color=KITred,line width = 1pt, solid, mark=+, mark size=2.5pt, mark options={solid}]
table[col sep=comma]{
0.0, 4.45e-01
0.5, 3.19e-01
1.0, 1.93e-01
1.5, 1.04e-01
2.0, 4.53e-02
2.5, 1.71e-02
3.0, 4.98e-03
3.5, 1.40e-03
4.0, 3.13e-04
4.5, 6.58e-05
};
\addlegendentry{SCL-$32$};

\addplot[color=KITgreen,line width = 1pt, dashed, mark=o, mark size=2pt, mark options={solid}]
table[col sep=comma]{
0.0, 5.27e-01
0.5, 3.76e-01
1.0, 2.50e-01
1.5, 1.44e-01
2.0, 6.74e-02
2.5, 2.69e-02
3.0, 8.55e-03
3.5, 2.44e-03
4.0, 5.57e-04
4.5, 1.12e-04
5.0, 1.83e-05
};
\addlegendentry{ScED-$2$-SCL-$8$};

\addplot[color=KITpalegreen,line width = 1pt, dashed, mark=x, mark size=2.5pt, mark options={solid}]
table[col sep=comma]{
0.0, 4.68e-01
0.5, 3.25e-01
1.0, 2.04e-01
1.5, 1.15e-01
2.0, 4.90e-02
2.5, 1.96e-02
3.0, 5.43e-03
3.5, 1.58e-03
4.0, 3.31e-04
4.5, 6.97e-05
};
\addlegendentry{ScED-$2$-SCL-$16$};

\addplot[color=black,line width = 1pt, solid, mark size=2.5pt, mark options={solid}]
table[col sep=comma]{
0.00, 3.643e-01
0.50, 2.318e-01
1.00, 1.496e-01
1.50, 8.415e-02
2.00, 3.886e-02
2.50, 1.392e-02
3.00, 3.995e-03
3.50, 1.128e-03
4.00, 2.720e-04
4.50, 6.565e-05
};
\addlegendentry{ML};

\coordinate (spypoint) at (axis cs:3.75,0.001);
			\coordinate (spyviewer) at (axis cs:3.95,0.075);	
			\spy[width=2.25cm,height=1.5cm, thin, spy connection path={\draw(tikzspyonnode.south west) -- (tikzspyinnode.south west);\draw (tikzspyonnode.south east) -- (tikzspyinnode.south east);
			\draw (tikzspyonnode.north west) -- (tikzspyinnode.north west);\draw (tikzspyonnode.north east) -- (intersection of  tikzspyinnode.north east--tikzspyonnode.north east and tikzspyinnode.south east--tikzspyinnode.south west);
			;}] on (spypoint) in node at (spyviewer);
		\coordinate (a) at ($(axis cs:-10.8/1.4,-0.12)+(spyviewer)$);
		\coordinate[label={[font=\small,text=black]right:$10^{-2}$}] (b) at ($(axis cs:+10.8/1.4,-0.12)+(spyviewer)$);
\end{axis}
\end{tikzpicture}
    \vspace*{-1.8em}
    \caption{Decoder performances for a 5G polar code $\mathcal{C}_T(64, 32)$  employing the \ac{CRC}-$6$ polynomial $0$x$03$ as a \PTA.}
    \label{fig:64_32_EDcomparison}
\end{figure}

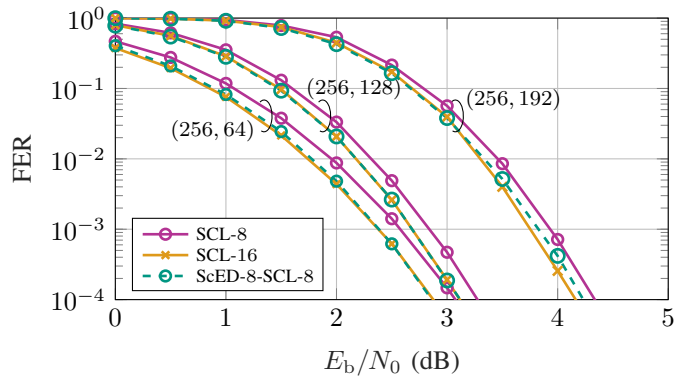
\begin{figure}[t]
    \centering
    \begin{tikzpicture}
\begin{axis}[
width=\columnwidth,
height=\columnwidth*0.6,
xmin=0,
xmax=5.0,
xlabel style={font=\color{white!15!black}},
xlabel={$E_{\mathrm{b}}/N_0$ (dB)},
ymode=log,
ymin=1e-04,
ymax=1,
ylabel style={font=\color{white!15!black}},
ylabel={FER},
axis background/.style={fill=white},
xmajorgrids,
ymajorgrids
]

\addplot[color=KITpurple!80,line width = 1pt, dashdotted]
table{
0 5
};

\addplot[color=KITorange,line width = 1pt, solid]
table{
0 5
};

\addplot[color=KITgreen,line width = 1pt, dashed]
table{
0 5
};

\addplot[color=KITpurple!80,line width = 1pt, solid,mark=o,mark size=2pt, mark options={solid}]
table{
0   4.68e-01
0.5 2.76e-01
1   1.18e-01
1.5 3.78e-02
2   8.75e-03
2.5 1.41e-03
3   1.47e-04
3.5 1.14e-05
};
\label{SCL8}

\addplot[color=KITorange,line width = 1pt, solid, mark=x,mark size=2pt, mark options={solid}]
table{
0   3.76e-01   
0.5 1.95e-01
1   7.60e-02
1.5 2.15e-02
2   4.38e-03
2.5 6.21e-04
3   5.44e-05
};
\label{SCL16}

\addplot[color=KITgreen, line width = 1pt, dashed, mark=o,mark size=2pt, mark options={solid}]
table{
0   4.05e-01
0.5 2.09e-01
1   8.22e-02
1.5 2.44e-02
2   4.82e-03
2.5 6.20e-04
3   5.02e-05
};
\label{ScED}

\addplot[color=KITpurple!80,line width = 1pt, solid,mark=o,mark size=2pt, mark options={solid}]
table{
0   8.33e-01
0.5 6.17e-01
1   3.54e-01
1.5 1.31e-01
2   3.33e-02
2.5 4.90e-03
3   4.71e-04
3.5 2.81e-05
};

\addplot[color=KITorange, line width = 1pt, solid, mark=x,mark size=2.5pt, mark options={solid}]
table{
0   7.93e-01
0.5 5.60e-01
1   2.87e-01
1.5 9.70e-02
2   2.08e-02
2.5 2.61e-03
3   1.91e-04
3.5 1.14e-05
};

\addplot[color=KITgreen,line width = 1pt, dashed, mark=o,mark size=2.5pt, mark options={solid}]
table{
0   7.82e-01
0.5 5.46e-01
1   2.82e-01
1.5 9.30e-02
2   2.08e-02
2.5 2.64e-03
3   1.87e-04
3.5 1.35e-05
};

\addplot[color=KITpurple!80,line width = 1pt, solid,mark=o,mark size=2pt, mark options={solid}]
table{
0   9.99e-01
0.5 9.88e-01
1   9.42e-01
1.5 7.85e-01
2   5.35e-01
2.5 2.15e-01
3   5.66e-02
3.5 8.54e-03
4   7.15e-04
4.5 3.82e-05
};

\addplot[color=KITorange, line width = 1pt, solid, mark=x,mark size=2.5pt, mark options={solid}]
table{
0   9.98e-01
0.5 9.85e-01
1   9.18e-01
1.5 7.47e-01
2   4.46e-01
2.5 1.72e-01
3   3.95e-02
3.5 4.01e-03
4   2.56e-04
4.5 1.43e-05
};

\addplot[color=KITgreen,line width = 1pt, dashed, mark=o,mark size=2.5pt, mark options={solid}]
table{
0   9.97e-01
0.5 9.79e-01
1   9.16e-01
1.5 7.23e-01
2   4.26e-01
2.5 1.65e-01
3   3.80e-02
3.5 5.17e-03
4   4.18e-04
4.5 1.91e-05
};

\node [draw,fill=white, font=\scriptsize, anchor=south west] at (rel axis cs: 0.03,0.03) {\shortstack[l]{
\ref{SCL8} SCL-$8$ \\
\ref{SCL16} SCL-$16$ \\
\ref{ScED} ScED-$8$-SCL-$8$}};

\draw (axis cs: 3.05, 0.025) arc(-110:120:0.3em and 0.6em);
\draw (axis cs: 1.85, 0.025) arc(-110:120:0.3em and 0.6em);
\draw (axis cs: 1.32, 0.025) arc(-110:120:0.3em and 0.6em);

\node[] at (axis cs: 3.6, 0.07){\footnotesize$(256,192)$};
\node[] at (axis cs: 2.16, 0.098){\footnotesize$(256,128)$};
\node[] at (axis cs: 0.88, 0.025){\footnotesize$(256,64)$};

\end{axis}
\end{tikzpicture}
    \vspace*{-1.8em}
    \caption{Decoder performances for 5G polar codes $\mathcal{C}_T(N, k)$, where ${k\in\{64,128,192\}}$, employing the \ac{CRC}-$11$ polynomial $0$x$621$.}
    \label{fig:256_128_FER}
\end{figure}

\section{Conclusion}
In this work, using a unified \ac{PT} framework, we have introduced a novel kind of \ac{PT} which generates polar subcodes and, thus, enables \ac{SCED} for polar codes. 
We have optimized the selection of subcodes for \ac{SCED} and compared the optimized ensembles for different parameters to \ac{SC} and \ac{SCL} decoding.
Our results show that, especially for shorter block lengths, \ac{SCED} offers promising performance while allowing for hardware-efficient implementation. 
Alternatively, \ac{SCED} can reduce latency compared to \ac{SCL} decoders of similar performance when hardware requirements are not as strict.

\end{document}